\documentclass[aps,%
onecolumn,%
floatfix,%
pra,%
amsfonts,%
groupedaddress,%
superscriptaddress]{revtex4-2}%
\usepackage[utf8]{inputenc}
\usepackage[T1]{fontenc}
\usepackage{mathrsfs}       
\usepackage{mathtools}
\usepackage{amsmath}
\usepackage{amsfonts}       
\usepackage[titletoc,title]{appendix}
\usepackage{algorithm,algorithmic}

\usepackage{bbm}            
\usepackage{bm}             
\usepackage{amsthm}         
\usepackage{makecell}
\usepackage[dvipsnames]{xcolor}
\definecolor{beamer@blendedblue}{rgb}{0.2,0.2,0.7}
\usepackage{times}
\usepackage{caption}
\usepackage{subfigure}
\usepackage{booktabs} 
\usepackage{multirow} 

\usepackage[colorlinks=true,%
linkcolor=beamer@blendedblue,%
bookmarks=true,%
breaklinks=true,%
filecolor=beamer@blendedblue,%
anchorcolor=yellow,%
citecolor=beamer@blendedblue,%
urlcolor=beamer@blendedblue,%
pdfauthor={Congcong Zheng, Ping Xu, Kun Wang and Zaichen Zhang},%
pdfsubject={GHZ-W Genuinely Entangled Subspace Verification with Adaptive Local Measurements},%
CJKbookmarks=true]{hyperref}

\usepackage{float}    
\usepackage{pifont}

\newtheorem{definition}{Definition}
\newtheorem{proposition}[definition]{Proposition}

\newtheorem{lemma}[definition]{Lemma}

\mathchardef\ordinarycolon\mathcode`\:
\mathcode`\:=\string"8000
\def\vcentcolon{\mathrel{\mathop\ordinarycolon}}
\begingroup \catcode`\:=\active
  \lowercase{\endgroup
  \let :\vcentcolon
  }

\DeclareFontFamily{U}{mathx}{\hyphenchar\font45}
\DeclareFontShape{U}{mathx}{m}{n}{<-> mathx10}{}
\DeclareSymbolFont{mathx}{U}{mathx}{m}{n}
\DeclareMathAccent{\widebar}{0}{mathx}{"73}

\newcommand{\wh}[1]{\widehat{#1}}

\newcommand{\ket}[1]{\vert{#1}\rangle}
\newcommand{\bra}[1]{\langle{#1}\vert}
\newcommand{\ketbra}[1]{\vert{#1}\rangle\!\langle{#1}\vert}
\newcommand{\braket}[2]{\langle #1\vert #2\rangle}
\newcommand\proj[1]{\vert{#1}\rangle\!\langle{#1}\vert}

\newcommand{\density}[1]{\mathscr{D}(#1)}

\DeclareMathOperator{\tr}{Tr}  
\newcommand{\1}{\mathbbm{1}}
\newcommand{\ox}{\otimes}

\makeatletter
\newsavebox{\@brx}
\newcommand{\llangle}[1][]{\savebox{\@brx}{\(\m@th{#1\langle}\)}%
  \mathopen{\copy\@brx\kern-0.5\wd\@brx\usebox{\@brx}}}
\newcommand{\rrangle}[1][]{\savebox{\@brx}{\(\m@th{#1\rangle}\)}%
  \mathclose{\copy\@brx\kern-0.5\wd\@brx\usebox{\@brx}}}
\makeatother

\newcommand*{\cM}{\mathcal{M}}

\newcommand*{\cV}{\mathcal{V}}

\newcommand{\bS}{\mathbb{S}}

\definecolor{wildstrawberry}{rgb}{1.0, 0.26, 0.64}
\definecolor{googleblue}{HTML}{4285F4}
\definecolor{googlered}{HTML}{DB4437}
\definecolor{googleyellow}{HTML}{F4B400}
\definecolor{googlegreen}{HTML}{0F9D58}
\definecolor{klevinblue}{HTML}{002FA7}
\definecolor{tiffanyblue}{HTML}{0ABAB5}

\begin{document}

\newcommand{\thetitle}{{GHZ-W Genuinely Entangled Subspace Verification with Adaptive Local Measurements}}
\title{\thetitle}

\author{Congcong Zheng}%
\affiliation{State Key Lab of Millimeter Waves, Southeast University, Nanjing 211189, China}%
\affiliation{Frontiers Science Center for Mobile Information Communication and Security, Southeast University, Nanjing 210096, People's Republic of China}%

\author{Ping Xu}%
\affiliation{Institute for Quantum Information \& State Key Laboratory of High Performance Computing, %
College of Computer Science and Technology, National University of Defense Technology, %
Changsha 410073, China}%
\affiliation{Hefei National Laboratory, Hefei 230088, China}%

\author{Kun Wang}
\thanks{Corresponding author: \href{nju.wangkun@gmail.com}{nju.wangkun@gmail.com}}%
\affiliation{Institute for Quantum Information \& State Key Laboratory of High Performance Computing, %
College of Computer Science and Technology, National University of Defense Technology, %
Changsha 410073, China}%

\author{Zaichen Zhang}%
\thanks{Corresponding author: \href{zczhang@seu.edu.cn}{zczhang@seu.edu.cn}}%
\affiliation{Frontiers Science Center for Mobile Information Communication and Security, Southeast University, Nanjing 210096, People's Republic of China}%
\affiliation{Purple Mountain Laboratories, Nanjing 211111, People's Republic of China}%
\affiliation{National Mobile Communications Research Laboratory, Southeast University, Nanjing 210096, China}%

\begin{abstract}
Genuinely entangled subspaces (GESs) are valuable resources in quantum information science. 
Among these, the three-qubit GHZ-W GES,
spanned by the three-qubit Greenberger-Horne-Zeilinger (GHZ) and W states, 
is a universal and crucial entangled subspace resource for three-qubit systems. 
In this work, we develop two adaptive verification strategies, the XZ strategy and the rotation strategy, for the three-qubit GHZ-W GES
using local measurements and one-way classical communication. 
These strategies are experimentally feasible, efficient and possess a concise analytical expression for the sample complexity of the rotation strategy, 
which scales approximately as $2.248\epsilon^{-1}\ln\delta^{-1}$, where $\epsilon$ is the infidelity and $1-\delta$ is the confidence level. 
Furthermore, we comprehensively analyze the two-dimensional two-qubit subspaces and classify them into three distinct types, 
which include unverifiable entangled subspaces, revealing intrinsic limitations in local verification of entangled subspaces.
\end{abstract}
\date{\today}
\maketitle

\section{Introduction}

Quantum entanglement, a fundamental aspect of quantum physics, has garnered significant attention in the realm of quantum information science~\cite{horodecki2009quantum}. 
Among the most prominent examples of multipartite entanglement are the Greenberger-Horne-Zeilinger (GHZ) and W states, which serve as paradigmatic instances~\cite{horodecki2009quantum, horodecki2024multipartite, chen2024quantumb}. 
Three-qubit entanglement is notably classified into two distinct types, represented by the GHZ and W states~\cite{dur2000three, acin2001classification, walther2005local}.  
A key area of research in multipartite entanglement focuses on subspaces composed entirely of entangled states, known as the completely entangled subspaces (CESs).
These subspaces have proven invaluable in applications such as quantum error correction~\cite{alsina2021absolutely, gour2007entanglement, huber2020quantuma, raissi2018optimal} and quantum cryptography~\cite{shenoy2019maximally}.
A particularly important class of completely entangled subspaces is the genuinely entangled subspace (GES), which consists solely of genuinely multipartite entangled states~\cite{demianowicz2022universal, baccari2020deviceindependent, makuta2021selftesting, demianowicz2018unextendible, demianowicz2021simple, demianowicz2019entanglement}.
A notable example of a GES is the GHZ-W subspace, spanned by the GHZ and W states~\cite{makuta2021selftesting}. 
This subspace is also an entangled symmetric subspace, where symmetry plays a fundamental role in quantum information theory \cite{li2020selftesting, martin2010multiqubit}, underlying many key applications. 
The GHZ-W subspace has garnered significant attention due to its foundational role in quantum teleportation~\cite{park2008role, chakrabarty2010teleportationa} and its broader implications in multipartite entanglement theory~\cite{lohmayer2006entangleda, wang2015classifying, xie2023evidence, faujdar2023nonlocality}. 
Notably, the three-qubit GHZ-W subspace serves as a universal resource for three-qubit entangled symmetric states \cite{zheng2020deterministic}, 
which has been extensively studied in~\cite{acin2001classification, lohmayer2006entangleda, li2020selftesting}. 

However, experimentally constructing three-qubit entangled symmetric resources remains challenging due to the pervasive influence of quantum noise. 
Consequently, accurately detecting the entanglement of the GHZ-W subspace has become a critical task in quantum information science. 
Quantum tomography, the standard method for characterizing entire quantum systems, provides comprehensive insight but is highly resource intensive~\cite{Hffner2005,PhysRevLett.129.133601}.
To address this limitation, numerous resource-efficient methods have been developed to certify quantum systems~\cite{eisert2020quantum, kliesch2021theory, huang2020predicting, elben2020crossplatforma, elben2023randomized, zheng2024crossplatform}. 
Notably, quantum state verification~\cite{pallister2018optimala, wang2019optimala, hayashi2009group, zhu2019optimal, yu2019optimal, zhu2019efficienta, yu2022statisticala, li2020optimal, han2021optimal, dangniam2020optimal, zhu2019efficientb, chen2025quantum, li2021verification, zhu2019general} aims to confirm whether quantum states are prepared as intended, with experimental implementations demonstrating its effectiveness~\cite{jiang2020towards, Zhang_2020, Xia_2022}. 
These verification strategies primarily use local operations and classical communication (LOCC)~\cite{Chitambar2014} to detect entangled states.
Naturally, certifying entangled subspaces, particularly GESs, has emerged as a critical task in quantum information science. 
However, entanglement certification within a subspace is inherently complex because of the structural intricacies of quantum subspaces. 
Recently, several approaches have been proposed to tackle this challenge, including subspace self-testing~\cite{baccari2020deviceindependent, makuta2021selftesting} and subspace verification~\cite{zheng2024efficient, chen2024quantuma, chen2023efficient, zhu2024efficient, fujii2017verifiable}.

In this work, building on the general framework of quantum subspace verification~\cite{zheng2024efficient, chen2024quantuma, chen2023efficient, zhu2024efficient, fujii2017verifiable}, 
our objective is to construct efficient strategies to verify the three-qubit GHZ-W subspace, which has not been considered in previous works. 
The verification or self-testing of GHZ states~\cite{hayashi2018secure, hayashi2022quantum, li2020optimal, han2021optimal}, W states~\cite{liu2019efficient}, and their superpositions~\cite{li2020selftesting} have been considered before.
However, the verification of the GHZ-W subspace is highly nontrivial, as it is inherently more complex than individual entangled states. 
To address this challenge, we construct verification strategies using one-way adaptive local measurements. 
Specifically, we measure one qubit and then adaptively adjust the second measurement conditioned on the measurement outcome.
This approach is intuitive, as it reduces the problem to verifying a much simpler two-qubit subspace.
We comprehensively analyze two-dimensional two-qubit subspaces and classify them into three distinct types: 
\emph{unverifiable}, \emph{verifiable}, and \emph{perfectly verifiable} subspaces. 
In particular, we prove that unverifiable subspaces cannot be certified by any LOCC strategy that relies on projective measurements. 
For the remaining two categories, we propose tailored verification strategies.
Building on these results, we develop two adaptive verification strategies for the three-qubit GHZ-W subspace, 
the XZ strategy and the rotation strategy, using local measurements and one-way classical communication.
The XZ strategy requires only four measurement settings.
The rotation strategy uses ten measurement settings, but achieves higher efficiency than the XZ strategy.
Notably, our strategies are experimentally friendly, requiring only LOCC, making them suitable for current quantum systems. 
Furthermore, our strategies offer new tools for analyzing the entanglement structure of mixed states within entangled subspace, contributing to the broader goal of certifying entanglement in realistic, noisy quantum systems.

\textbf{Related Works.}
For a self-contained review, we summarize the related works as follows:
\begin{itemize}
\item \emph{State Verification.} 
Sample-optimal strategies have been developed for a variety of pure quantum states, including two-qubit states~\cite{pallister2018optimala, wang2019optimala}, bipartite maximally entangled states~\cite{hayashi2009group, zhu2019optimal}, general bipartite pure states~\cite{yu2019optimal}, GHZ states~\cite{li2020optimal, han2021optimal}, stabilizer states~\cite{dangniam2020optimal}, and antisymmetric basis states~\cite{li2021verification}.
Efficient protocols have also been proposed for verifying other pure states such as hypergraph states~\cite{zhu2019efficientb}, Dicke and W states~\cite{liu2019efficient}, phased Dicke states~\cite{li2021verification}, Affleck-Kennedy-Lieb-Tasaki states~\cite{chen2023efficient}, and even arbitrary entangled pure states~\cite{liu2023efficient}.
Many of these protocols leverage adaptive measurements~\cite{liu2019efficient, yu2019optimal, li2021verification, liu2023efficient}.
\item \emph{Subspace Verification.} 
The concept of verifying quantum subspaces was first introduced in~\cite{fujii2017verifiable}, aiming to verify the output of fault-tolerant quantum computation in measurement-based quantum computation. 
It was later extended to verify ground states of local 
Hamiltonians~\cite{zhu2024efficient, chen2023efficient}. 
More recently, subspace verification has been applied to verify the implementation of quantum error correction codes~\cite{zheng2024efficient, chen2024quantuma}. 
\end{itemize}

The remainder of this paper is organized as follows. 
In Section~\ref{sec:subspace verification}, we provide a brief overview of the subspace verification framework. 
Section~\ref{sec:verification of two-qubit subspace} focuses on the classification and verification of two-qubit subspaces. 
In Section~\ref{eq:ghz-w subspace verification}, we present two efficient verification strategies for the three-qubit GHZ-W subspace.


\section{Subspace verification}
\label{sec:subspace verification}
Let us first review the framework for the statistical verification of the quantum subspace~\cite{zheng2024efficient, chen2024quantuma, chen2023efficient, zhu2024efficient, fujii2017verifiable}.
Suppose that our objective is to prepare target states within a subspace $\cV$, but in practice
we obtain a sequence of states $\sigma_1, \cdots, \sigma_N$ in $N$ runs. 
Let $\density{\cV}$ be the set of density operators acting on $\cV$
and $\Pi$ be the projector onto $\cV$. 
Our task is to distinguish between the following two cases:
\begin{enumerate}
    \item \textbf{Good}: 
    for all $i\in[N]$, $\tr[\Pi\sigma_i] = 1$;
    \item \textbf{Bad}: 
    for all $i\in[N]$, $\tr[\Pi\sigma_i] \leq 1 - \epsilon$ for some fixed $\epsilon$. 
\end{enumerate}
To achieve this, assume that we have access to a set of POVM elements $\cM$. 
Define a probability mass $\mu: \cM \to [0, 1]$, satisfying $\sum_{M\in\cM} \mu(M) = 1$. 
For each state, we select a POVM element $M\in\cM$ with probability $\mu(M)$ 
and perform the corresponding POVM with two results $\{M, \1-M\}$, where the $M$ outputs ``pass'' and the $\1-M$ outputs ``fail''. 
The operator $M$ is called a \emph{test operator}. 
To ensure that all states in the target subspace pass the test, we require $\tr[M\rho]=1$ for all $\rho\in\density{\cV}$ and $M\in\cM$.
The sequence of states passes the verification procedure if all outcomes are ``pass''.

Mathematically, we can characterize the verification strategy by the \emph{verification operator}, defined as 
\begin{align}
    \Omega = \sum_{M\in\cM} \mu(M) M. 
\end{align}
If $\tr[\Pi\sigma_i]$ is upper bounded by $1-\epsilon$, the maximal probability that $\sigma_i$ passes each test is~\cite{fujii2017verifiable, chen2023efficient, zhu2024efficient}
\begin{align}
    \max_{\sigma:\tr[\Pi\sigma]\leq1-\epsilon}\tr[\Omega\sigma]
    = 1 - (1-\lambda_{\max}(\wh{\Omega}))\epsilon,
\end{align}
where $\wh{\Omega}:=(\1-\Pi)\Omega(\1-\Pi)$ is the 
projected effective verification operator and $\lambda_{\max}(X)$ denotes the maximum eigenvalue of the Hermitian operator $X$.
If the states are independently prepared, the probability of passing $N$ tests is bounded by~\cite{hayashi2009group}
\begin{align}
    \prod_{i=1}^N \tr[\Omega\sigma_i] \leq (1 - \nu(\Omega)\epsilon)^N, 
\end{align}
where $\nu(\Omega) := 1 - \lambda_{\max}(\wh{\Omega})$ is the \emph{spectral gap}. 
To achieve a confidence level of $1-\delta$, the minimum required number of state copies is given by 
\begin{align}\label{eq:sample-complexity}
    N(\Omega) 
    = \left\lceil\frac{1}{\ln(1-\nu(\Omega)\epsilon)^{-1}}\ln\frac{1}{\delta}\right\rceil
    \approx \left\lceil\frac{1}{\nu(\Omega)}\times\frac{1}{\epsilon}\ln\frac{1}{\delta}\right\rceil. 
\end{align}
This equality provides a guide for the construction of efficient verification strategies by maximizing $\nu(\Omega)$. 
If there is no restriction on measurements, 
the globally optimal strategy is achieved by simply performing 
the projective measurement $\{\Pi,\1-\Pi\}$, 
which produces $\nu(\Pi) = 1$ and
\begin{align}\label{eq:global-strategy}
N_G(\Pi) = \left\lceil\frac{1}{\epsilon}\ln\frac{1}{\delta}\right\rceil.
\end{align}
 
However, implementing the globally optimal strategy requires highly entangled measurements 
if the target subspace is genuinely entangled, which are experimentally challenging. 
Consequently, we focus on the verification of the subspace under the locality constraint, 
where each test operator $M$ is a local projector. 
These strategies significantly improve experimental feasibility while still enabling efficient verification of the target subspace.

\section{Two-qubit subspace verification}
\label{sec:verification of two-qubit subspace}

For the three-qubit target subspace to be verified, measuring one qubit naturally projects the remaining two qubits into a two-qubit subspace, 
conditioned on the measurement outcome. 
Therefore, we begin by discussing the verification of two-qubit subspaces. 
Remarkably, two-dimensional two-qubit subspaces can be classified into three distinct types, each characterized by its unique properties. 
The basic properties of two-dimensional subspaces of two-qubits were clarified in~\cite{englert2002kinematics}. 
In particular, we demonstrate that there exists a two-dimensional two-qubit subspace that cannot be verified by any strategy using LOCC.
For the remaining two categories, we propose tailored verification strategies and elaborate their corresponding efficiencies.

\subsection{When a two-qubit subspace is verifiable?}

First, we identify the types of subspaces that can be verified. 
Intuitively, a subspace is verifiable if its complement subspace can be spanned by product states; 
otherwise, it cannot be verified.
Now, consider a two-dimensional subspace $\cV$ of a two-qubit Hilbert space, its complement subspace is denoted by $\cV^\bot$.
A key insight from quantum state verification is that product states in the complement subspace should be identified first~\cite{pallister2018optimala}. 
In particular, two-qubit product states can be efficiently verified using the following method.
Any two-qubit pure state $\ket{\psi}$ can be uniquely represented by a $2\times 2$ matrix $\psi$ as~\cite{wootters1998entanglementa,duan2009distinguishability}: 
\begin{align}
    \ket{\psi} = \1 \ox \psi \ket{\Phi}, 
\end{align}
where $\ket{\Phi} = (\ket{00} + \ket{11}) / \sqrt{2}$ is the maximally entangled state. 
The \emph{concurrence} of $\ket{\psi}$, defined as $C(\ket{\psi}) := |\det(\psi)| \leq 1$, quantifies its entanglement~\cite{wootters1998entanglementa}, where $\det(A)$ is the determinant of the square matrix $A$. 
Specially, if $C(\ket{\psi}) = 0$, $\ket{\psi}$ is a product state.
This criterion enables a direct method for determining the number of product states in $\cV$.
Let $\ket{\alpha}$ and $\ket{\beta}$ be two linearly independent states spanning $\cV$.
Then all states in $\cV$ can be expressed as $\ket{\alpha} + \lambda \ket{\beta}$ (unnormalized) for some $\lambda \in \mathbb{C}$.
By solving the equation $\det(\alpha+\lambda\beta)=0$, we can find all product states in $\cV$.
The following lemma shows that the number of distinct product states contained in $\cV$ is equal to that in $\cV^\bot$.

\begin{lemma}
\label{lemma:number of product states}
For any two-dimensional subspace $\mathcal{V}$ of a two-qubit Hilbert space, the number of distinct product states in $\mathcal{V}$ \emph{equals} that in $\mathcal{V}^\perp$, and is either $1$, $2$, or $+\infty$.
\end{lemma}

\begin{proof}
Here, we disregard global phase factors. 
Let $N$ be the number of distinct product states in $\cV$. 
Since the maximum dimension of a CES in two-qubit Hilbert space is $1$~\cite{demianowicz2020approach, parthasarathy2004maximal} and $\cV$ is two-dimensional, it holds that $N \geq 1$.

First, suppose that there are at least two distinct product states in $\cV$, i.e., $N \geq 2$. 
Without loss of generality, using the Schmidt decomposition, 
we can write these two states as $\ket{00}$ and $\ket{ab}$, where $\ket{a}, \ket{b}$ are single-qubit states. 
In this case, there must be two product states in $\cV^\bot$, given by $\ket{1\bar{b}}$ and $\ket{\bar{a}1}$, 
where $\ketbra{\bar{a}}+\ketbra{a} = \1$ and similarly for $\ket{\bar{b}}$. 
Then, consider the following cases:
\begin{enumerate}
\item If $\ket{a}=\ket{0}$ or $\ket{b} = 0$, there are \emph{infinitely} many product states in $\cV$, e.g., $\cV = {\rm span}\{\ket{00}, \ket{01}\}$ or $\cV = {\rm span}\{\ket{00}, \ket{10}\}$. 
In both cases, the orthogonal complement $\cV^\bot$ also contains \emph{infinitely} many product states,
with $\{\ket{10}, \ket{11}\}$ or $\{\ket{01}, \ket{11}\}$ being orthogonal product bases in this subspace, respectively.
\item Otherwise, $\ket{00}$ and $\ket{ab}$
($\ket{1\bar{b}}$ and $\ket{\bar{a}1}$) are the only product states in $\cV$ ($\cV^\bot$). 
\end{enumerate}
The above constructive proof reveals that, if $\cV$ has at least two distinct product states,
then the number of distinct product states in $\mathcal{V}$ equals that in $\mathcal{V}^\perp$, 
and is either $2$ or $+\infty$.

Then, we consider the case $N=1$, where there is only one product state in $\cV$.
Again, since the maximum dimension of a CES two-qubit Hilbert space is $1$ and $\cV^\perp$ is two-dimensional,
it holds that $\cV^\perp$ must have at least one product state.
We prove by contradiction that $\cV^\perp$ has exactly one product state.
Assume that $\cV^\bot$ contains at least two distinct product states. 
It follows from the proof for the case $N\geq2$ that, 
there are $2$ or $+\infty$ distinct product states in $\cV$, 
thus contradicting the assumption that $N=1$. 
We are done.
\end{proof}

We now show that whether $\cV$ is verifiable is determined by the number of distinct product states in $\cV$. 
Based on Lemma~\ref{lemma:number of product states}, 
we classify all two-dimensional subspaces of a two-qubit Hilbert space into three different types:
verifiable, perfectly verifiable, and unverifiable subspaces.

\paragraph*{\textbf{Verifiable and perfectly verifiable subspaces.}} 
If $\cV$ only contains two distinct product states, 
then, by Lemma~\ref{lemma:number of product states}, we can span $\cV^\bot$ using two product states. 
This implies that we can verify the subspace with two test operators:
\begin{align}
    M_i = \1 - \ketbra{\tau_i}, \quad i = 0,1, 
\end{align}
where $\ket{\tau_i}$ are the product states in $\cV^\bot$. 
We call such a subspace $\cV$ a \emph{verifiable subspace}.
Specifically, if these are two orthogonal product states, i.e., $\braket{\tau_0}{\tau_1} = 0$, we can verify the subspace with only one test operator of the form: 
\begin{align}
    M = \1 - (\ketbra{\tau_0} + \ketbra{\tau_1}). 
\end{align}
In this case, we refer to $\cV$ as a \emph{perfectly verifiable subspace}.
Obviously, a subspace with infinitely many product states is also a perfectly verifiable subspace, 
as it is evident from the constructive proof of Lemma~\ref{lemma:number of product states}
that we can always find two orthogonal product states in the corresponding complement subspace.

For example, the subspace spanned by $\ket{00}$ and $(\ket{0} + \ket{1})^{\ox 2}/2$ is a verifiable subspace, since its complement subspace contains two non-orthogonal product states $\ket{1}\ox(\ket{0} - \ket{1})/\sqrt{2}$ and $(\ket{0} - \ket{1})\ox\ket{1}/\sqrt{2}$. 
The subspace spanned by $(\ket{00} + \ket{11})/\sqrt{2}$ and $(\ket{00} - \ket{11})/\sqrt{2}$ is a perfectly verifiable subspace, since its complement subspace contains two orthogonal product states $\ket{01}$ and $\ket{10}$. 

\paragraph*{\textbf{Unverifiable subspace.}} 
On the other hand, if there is only one product state in $\cV$, then, by Lemma~\ref{lemma:number of product states}, 
we cannot span $\cV^\bot$ with product states . 
This type of subspace is called an \emph{unverifiable subspace}. 
In this case, we can only reject this product state in the test, and the corresponding test operator is:
\begin{align}
    M = \1 - \ketbra{\tau}, 
\end{align}
where $\ket{\tau}$ is the only product state in $\cV^\bot$. 
Mathematically, the corresponding verification operator is:
\begin{align}
    \Omega_{u} = M = \1 - \ketbra{\tau}. 
    \label{eq:verification operator of unverifiable subspace}
\end{align}
We have $\nu(\Omega_u) = 0$, which means this strategy is inevitably fooled by a state $\vert\tau'\rangle$, 
where $\vert\tau'\rangle$ is an entangled state in the complement subspace and 
$\langle\tau\vert\tau'\rangle = 0$. 
Therefore, there is no verification strategy based on LOCC and projective measurements for an unverifiable subspace. 
For example, the subspace spanned by $(\ket{00} + \ket{11})/\sqrt{2}$ and $\ket{01}$ is an unverifiable subspace, 
since its complement subspace only contains one orthogonal product state $\ket{10}$. 
The verification operator defined in Eq.~\eqref{eq:verification operator of unverifiable subspace} will be fooled by $(\ket{00} - \ket{11})/\sqrt{2}$. 

\subsection{Verification strategy}

With the above classification, we design a verification strategy tailored to (perfectly) verifiable subspaces and analyze their corresponding spectral gaps. 

\paragraph*{\textbf{Perfectly verifiable subspace}.}
For a perfectly verifiable subspace, we construct a two-outcomes POVM $\{\ketbra{\tau_0} + \ketbra{\tau_1}, \1 - \ketbra{\tau_0} - \ketbra{\tau_1}\}$, 
where $\ket{\tau_i}~(i=0,1)$ are two \emph{orthogonal} product states in the target subspace. 
We pass the state with the result corresponding to $\ketbra{\tau_0} + \ketbra{\tau_1}$. 
Mathematically, the corresponding verification operator is given by 
\begin{align}
    \Omega_{p} = \ketbra{\tau_0} + \ketbra{\tau_1}. 
    \label{eq:verification operator of perfectly verifiable subspace}
\end{align}
Obviously, we have $\nu(\Omega_p) = 0$, which means no states from the complement subspace can pass this strategy. 
Therefore, to achieve a confidence level of $1-\delta$, it suffices to choose 
\begin{align}
    N(\Omega_p) = \left\lceil\frac{1}{\epsilon}\ln\frac{1}{\delta}\right\rceil. 
\end{align}
We can also refer to this kind of subspace as a \emph{local subspace}.

\paragraph*{\textbf{Verifiable subspace}.}
For a verifiable subspace, the strategy is slightly more complex than for other types.
It involves two POVMs: $\{\1-\ketbra{\tau_2}, \ketbra{\tau_2}\}$ and $\{\1-\ketbra{\tau_3}, \ketbra{\tau_3}\}$, 
where $\ket{\tau_i}~(i=2,3)$ are product states in the complement subspace. 
Each POVM is performed with probability $1/2$ and we pass the states with the result corresponding to the $\1 - \ketbra{\tau_i}~(i=2,3)$. 
Mathematically, the corresponding verification operator is given by 
\begin{align}
    \Omega_{v} = \1 - \frac{1}{2}(\ketbra{\tau_2}+\ketbra{\tau_3}).
    \label{eq:verification operator of verifiable subspace}
\end{align}
Although states in the complement subspace can pass each test individually, they cannot pass all tests with certainty. 
The spectral gap of this verification operator is given as follows. 

\begin{lemma}
\label{lem:spectral gap of verifiable subspace}
For a verifiable subspace, the spectral gap of the verification operator defined in Eq.~\eqref{eq:verification operator of verifiable subspace} is 
\begin{align}
\nu(\Omega_v) = \frac{1}{2}(1-|\braket{\tau_3}{\tau_2}|). 
\end{align}
where $\ket{\tau_i}~(i=2,3)$ are the product states in the complement subspace. 
\end{lemma}

\begin{proof}
The two (unnormalized) eigenstates of $\Omega_v$ in $\cV^\bot$ are 
\begin{align}
    \ket{\tau_2} + \frac{a}{|a|}\ket{\tau_3}, \quad 
    \ket{\tau_2} - \frac{a}{|a|}\ket{\tau_3}, 
\end{align}
where $a = \braket{\tau_3}{\tau_2}$.
The corresponding eigenvalues are $(1-|a|)/2$ and $(1+|a|)/2$, respectively. 
Using the definition of the spectral gap, we have 
\begin{align}
    \nu(\Omega_v) = 1 - \frac{1}{2}(1 + |a|) = \frac{1}{2}(1 - |a|). 
\end{align}
\end{proof}
Therefore, with the result of Lemma~\ref{lem:spectral gap of verifiable subspace}, it suffices to choose 
\begin{align}
    N(\Omega_v) 
    = \left\lceil\frac{2}{1-|\braket{\tau_2}{\tau_3}|}\times \frac{1}{\epsilon}\ln\frac{1}{\delta}\right\rceil, 
\end{align}
to achieve a confidence level of $1-\delta$.

\section{GHZ-W subspace verification} 
\label{eq:ghz-w subspace verification}

In this section, building on the results of two-qubit subspace verification, we propose two efficient strategies to verify 
the subspace $\cV_3:={\rm span}\{\ket{{\rm GHZ}},\ket{{\rm W}}\}$ spanned by the three-qubit GHZ and W states, where
\begin{subequations}
\label{eq:W and GHZ}
\begin{align}
    \ket{{\rm GHZ}} &:= \frac{1}{\sqrt{2}}(\ket{000} + \ket{111}), \\
    \ket{{\rm W}} &:= \frac{1}{\sqrt{3}}(\ket{001} + \ket{010} + \ket{100}). 
\end{align}
\end{subequations}
We call $\cV_3$ the \emph{three-qubit GHZ-W subspace}, which is genuinely multipartite entangled~\cite{makuta2021selftesting}.
Notably, most verification protocols for bipartite pure states, Dicke states, and W states known so far are based on adaptive measurements~\cite{liu2019efficient, yu2019optimal, li2021verification, liu2023efficient}.
Therefore, we begin by constructing multiple test operators based on one-way adaptive measurements. 
Subsequently, we propose two efficient verification strategies and conduct a detailed analysis of their sample complexities.

\subsection{One-way adaptive test operators} 
\label{sec:one-way adaptive test operators}

One-way adaptive measurements are widely used in the certification of quantum information and were first proposed in Section 6 of~\cite{hayashi2009group}, owing to their experimental feasibility and practicality.
This approach was further developed in subsequent works~\cite{hayashi2015verifiable, fujii2017verifiable, hayashi2018secure}.
We present a general subroutine to construct one-way adaptive measurements suitable for verifying $\cV_3$.
Specifically, we randomly measure a qubit in the Pauli basis $P\in\{X,Z\}$. 
Each measurement yields one of two possible outcomes, $+1$ and $-1$, 
corresponding to the positive and negative eigenspaces of $P$, respectively.
Depending on the measurement outcome, the remaining two qubits are projected into a two-qubit subspace spanned by two post-measurement states, 
called the \emph{post-measurement subspace}.
Table~\ref{table:post-measurement-states of W and GHZ} summarizes all possible post-measurement states for different measurement operators and outcomes.
Subsequently, based on the outcome of the first measurement, we apply the two-qubit subspace test operator. 
Therefore, the corresponding one-way adaptive test operators induced by $P$ are given by
\begin{align}
    M_{P} = P^+\ox M_{P}^+ + P^-\ox M_{P}^-. \label{eq: test operator XZ}
\end{align}
That is, if the outcome of $P$ is $+1$, we perform the two-qubit measurement associated with $M_{P}^+$.
Otherwise, we perform the measurement corresponding to $M_{P}^-$. 
Finally, the states that produce results consistent with $M_{P}$ pass the test.

\begin{table*}[!htbp]
\centering
\renewcommand{\arraystretch}{1.5} 
\setlength{\tabcolsep}{18pt}      
\setlength\heavyrulewidth{0.3ex}  
\begin{tabular}{@{}cccc@{}}
\toprule
\multicolumn{2}{c}{\bf First measurement} & \multicolumn{2}{c}{\bf Post-measurement states} \\ \midrule
Pauli & outcome & $\ket{{\rm GHZ}}$ &  $\ket{{\rm W}}$ \\\midrule[0.06ex]
\multirow{2}{*}{$Z$} & $+$ & $\ket{00}$ & $\frac{1}{\sqrt{2}}(\ket{01}+\ket{10})$ \\
 & $-$ & $\ket{11}$ & $\ket{00}$  \\ \midrule[0.06ex]
\multirow{2}{*}{$X$} & $+$ & $\frac{1}{\sqrt{2}}(\ket{00}+\ket{11})$ & $\frac{1}{\sqrt{3}}(\ket{01}+\ket{10}+\ket{00})$ \\
& $-$ &$\frac{1}{\sqrt{2}}(\ket{00}-\ket{11})$ & $\frac{1}{\sqrt{3}}(\ket{01}+\ket{10}-\ket{00})$ \\ 
\bottomrule
\end{tabular}
\caption{Post-measurement states for the three-qubit GHZ-W subspace 
spanned by $\{\ket{{\rm GHZ}}, \ket{{\rm W}}\}$.} 
\label{table:post-measurement-states of W and GHZ}
\end{table*}

Now, let us consider two concrete cases where $P$ is chosen to be the Pauli $Z$ or the $X$ measurements.
For the post-measurement subspace induced by the Pauli $Z$ measurement, 
if the outcome is ``$+$'' (``$-$''), the resulting two-qubit subspace is unverifiable (perfectly verifiable).
The test operators are given by  
\begin{align}
    M_{Z}^+ = \1 - \proj{11}, \qquad
    M_{Z}^- = \proj{00} + \proj{11}. 
\end{align}
The resulting one-way adaptive test operator induced by the $Z$ measurement thus has the form
\begin{align}
    M_{Z} = Z^+\ox M_{Z}^+ + Z^-\ox M_{Z}^- = \proj{0}\ox \left(\1 - \proj{11}\right) + \proj{1}\ox\left(\proj{00} + \proj{11}\right).
\end{align}
Actually, this one-way adaptive test operator can be implemented non-adaptively by performing the $Z$ measurements on each qubit.
The state is rejected if the measurement outcome contains exactly two ``$-$'' results.
Likewise, for the post-measurement subspace induced by the $X$ measurement, 
if the outcome is ``$+$'' (``$-$''), the resulting two-qubit subspace is perfectly verifiable (unverifiable). 
The test operators are given by
\begin{align}
    M_{X}^+ = \proj{x_+ x_+} + \proj{\bar{x}_+ \bar{x}_+}, \qquad
    M_{X}^- = \1 - \ketbra{x_- x_-'},
\end{align}
where the states are defined as
\begin{equation}
\begin{split}
    \ket{x_+} &= \cos\alpha\ket{0} + \sin\alpha\ket{1}, \\
    \ket{\bar{x}_+} &= \sin\alpha\ket{0} - \cos\alpha\ket{1}, \\
    \ket{x_-} &= \frac{\ket{0} + e^{i\frac{\pi}{3}}\ket{1}}{\sqrt{2}}, \\
    \ket{x_-'} &= \frac{\ket{0} + e^{-i\frac{\pi}{3}}\ket{1}}{\sqrt{2}}, 
\end{split}
\end{equation}
with $\alpha = \arctan (\sqrt{5}-1)/2$. The resulting one-way adaptive test operator induced by the $X$ measurement thus has the form 
\begin{align}
    M_{X} = \proj{+}\ox\left(\proj{x_+ x_+} + \proj{\bar{x}_+ \bar{x}_+}\right) + \proj{-}\ox\left(\1 - \ketbra{x_- x_-'}\right), 
\end{align}
where $\ket{+} = (\ket{0} + \ket{1}) / \sqrt{2}$ and $\ket{-} = (\ket{0} - \ket{1})/\sqrt{2}$ are two eigenstates of $X$, respectively.

To construct additional test operators beyond $M_X$ and $M_Z$, a general framework is necessary. 
An important observation from quantum state verification is that the local symmetry of the target subspace can be exploited to create more
test operators from current test operators~\cite{pallister2018optimala}. 
Specifically, if a product unitary $U$ satisfies $U\Pi_3 U^\dagger = \Pi_3$, where $\Pi_3$ is the projector of $\cV_3$, then $U$ is a local symmetry operator of $\cV_3$.
This symmetry also enables an analytical determination of the spectral gap, possibly optimizing performance.
Motivated by this observation, we identify the following two local symmetries of $\cV_3$:
\begin{enumerate}
\item \emph{Qubit permutations:}
\begin{align}
    V_\sigma 
    := \sum_{i_1,i_2,i_3} \ket{i_{\sigma^{-1}(1)}i_{\sigma^{-1}(2)}i_{\sigma^{-1}(3)}}\!\bra{i_1 i_2 i_3}, 
\end{align}
where $\sigma$ ranges over all elements of the symmetric group $\bS_3$; and 
\item \emph{Local unitaries}:
\begin{align}
    U_1 &:= R_{2\pi/3}\ox R_{2\pi/3}\ox R_{2\pi/3}, \\
    U_2 &:= R_{4\pi/3}\ox R_{4\pi/3}\ox R_{4\pi/3},
\end{align}
where $R_\phi := \ketbra{0} + e^{i\phi}\ketbra{1}$ and $U_2 = U_1^2$.
\end{enumerate}
One can check that $V_\sigma \Pi_3 V_\sigma^\dagger = \Pi_3$ and $U_i \Pi_3 U_i^\dagger = \Pi_3$ for $\sigma\in\bS_3$ and $i=1,2$. 
Using these two local symmetries of $\cV_3$, we can construct additional test operators.

Notice that $M_Z$ is invariant under the above local symmetries, 
so we focus on constructing additional test operators from $M_X$. 
First, we consider the qubit permutation symmetry. 
We define $M_{X, i}$ ($i=1,2,3$) as a set of new test operators, where an $X$ measurement is performed on the $i$-th qubit,
followed by a two-qubit verification based on the measurement result.
This construction takes advantage of the qubit permutation symmetry $V_\sigma$. 
Therefore, with this property, we can construct $3$ additional test operators.
Then, we use the local unitary symmetry.
We observe that $U_j^\dagger M_{X,i} U_j$ ($j=1,2$) are also valid one-way adaptive test operators, since the subspace $\cV_3$ is invariant under the local unitaries $U_j$. 
Physically, $U_j^\dagger M_{X,i} U_j$ corresponds to first applying the local rotation operator $U_j$ to the quantum state, 
followed by the test operator $M_{X,i}$. 
Consequently, we can construct a total of six additional test operators, given by $2\times 3 = 6$. 

To summarize, we build $10$ test operators for the GHZ-W subspace by applying local symmetries. 
We then present two verification strategies using these one-way adaptive test operators and assess their effectiveness.

\subsection{XZ strategy}

Here, we propose a verification strategy using the $4$ test operators constructed above, termed the \emph{XZ strategy}. 

\paragraph*{The strategy.}
In each round, we select a measurement $P\in\{X, Z\}$ according to a probability distribution $\mu(P)$, which will be optimized later. 
If the $Z$ measurement is chosen, we perform the test operator $M_Z$. 
Otherwise, we choose a qubit $i\in\{1,2,3\}$ uniformly at random to perform the test project $M_{X,i}$. 
Mathematically, the verification operator can be written as
\begin{align}
    \Omega_{\rm XZ} = \mu(Z) M_Z + \frac{1}{3}\mu(X)\sum_{i=1}^{3} M_{X,i}. 
\end{align}

\paragraph*{Performance analysis.}
It is challenging to analytically determine the optimal probability $\mu$ that maximizes the spectral gap of $\Omega_{\rm XZ}$. 
To address this, we numerically analyze the performance of the verification operator $\Omega_{\rm XZ}$.
We sample $\mu(Z)$ from $0$ to $1$ with a step size of $0.001$ and
compute the spectral gaps. The numerical results are presented in Fig.~\ref{fig:numerical results}(a)
and show that $\max \nu(\Omega_{\rm XZ}) \approx 0.262$ when $\mu(Z)\approx 0.424$. 
Therefore, to achieve a confidence level of $1-\delta$, the required number of state copies is given by 
\begin{align}\label{eq:XZ-strategy}
    N(\Omega_{\rm XZ}) 
    \approx \left\lceil3.817\times\frac{1}{\epsilon}\ln\frac{1}{\delta}\right\rceil.
\end{align}

\emph{Remark.}
The design of our XZ strategy is directly inspired by the seminal work of Hayashi and Morimae on quantum graph states verification~\cite{hayashi2015verifiable}. 
In their foundational protocol for verifiable measurement-only blind quantum computing, the client verifies graph states with $X$ and $Z$ measurements, thereby ensuring the correctness of the computation result. 
This pioneering approach was later adapted to verify the output of fault-tolerant quantum computation in the measurement-based model~\cite{fujii2017verifiable}. 
Building on these ideas, we use adaptive measurements based on $X$ and $Z$ measurements to verify the GHZ-W subspace.

\begin{figure}[!htbp]
    \centering
    \includegraphics[width=1\linewidth]{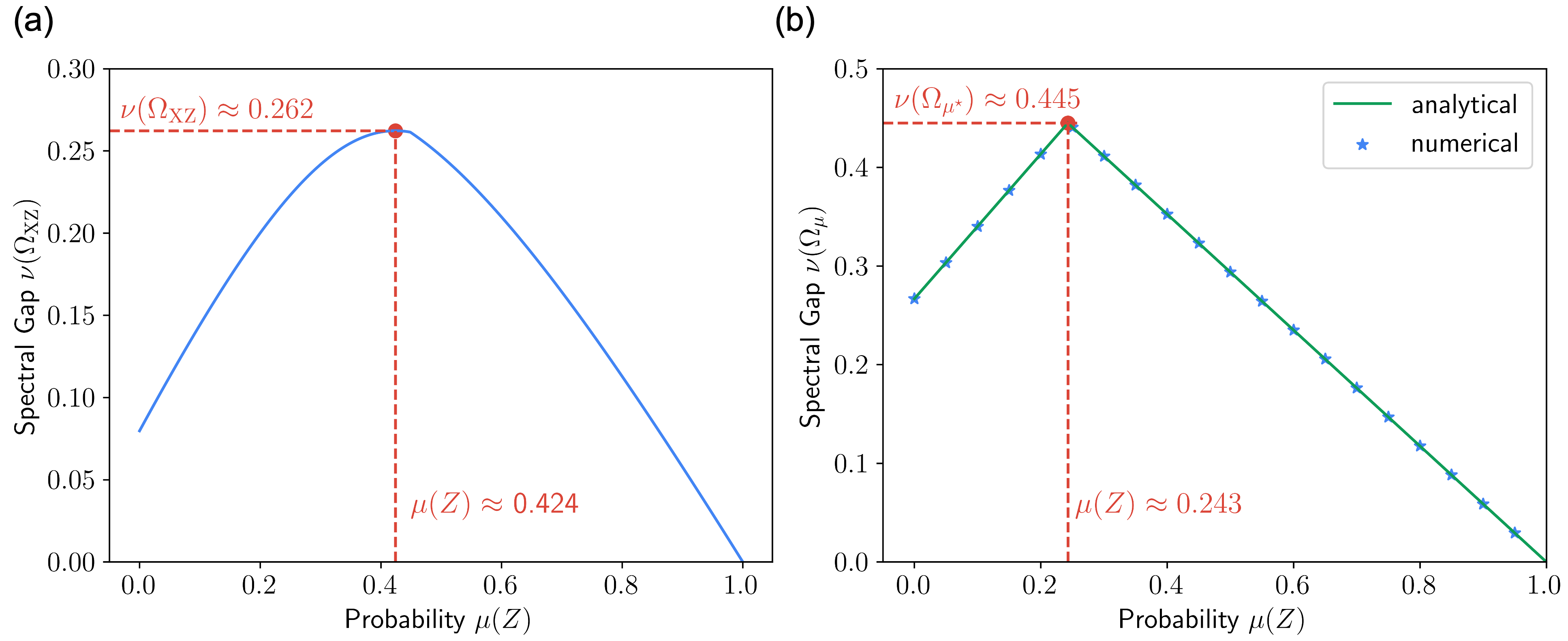}
    \caption{\raggedright 
    \textbf{(a)} Numerical result of the XZ strategy. 
    The probability $\mu(Z)$ is sampled from $0$ to $1$ with a step size of $0.001$, 
    and the spectral gaps are computed. 
    The maximal spectral gap is achieved when $\mu(Z)\approx 0.424$.
    \textbf{(b)} Analytical and numerical results of the rotation strategy. 
    The two green lines represent the functions $ \frac{47}{80}\mu(X)$ and $1 - \frac{11}{15}\mu(X)$, respectively. 
    The blue stars indicate the numerical results of the rotation strategy, where $\mu(X)$ is sampled from $0$ to $1$ with a step size of $0.05$. 
    The maximal spectral gap is achieved when $\mu(Z) = \frac{77}{317}\approx 0.243$.}
    \label{fig:numerical results}
\end{figure}

\subsection{Rotation strategy} 

The XZ strategy, while effective, lacks an analytical solution and exhibits a sample complexity approximately four times that of the globally optimal strategy. To address these limitations, we introduce the rotation strategy, utilizing the $10$ test operators constructed in Section~\ref{sec:one-way adaptive test operators}. This strategy, for which we derive an analytical performance, achieves a sample complexity approximately twice that of the globally optimal strategy.

\paragraph*{The strategy.}
In each round, we select a measurement $P\in\{X, Z\}$ according to a probability distribution $\mu(P)$, which will be optimized later. 
If the $Z$ measurement is chosen, we perform the test operator $M_Z$. 
Otherwise, we apply $U_1$, $U_2$ or $\emptyset$ (no unitary at all) uniformly at random, 
followed by choosing one qubit uniformly at random to carry out the test project $M_{X,i}$. 
Mathematically, this test operator can be written as
\begin{align}
    \widehat{M}_X := \frac{1}{3} \sum_{i=1}^3 M_{X,i}', 
\end{align}
where 
\begin{align}
    M'_{X,i} := \frac{1}{3}(M_{X,i} + U_1^\dagger M_{X,i} U_1 + U_2^\dagger M_{X,i} U_2). 
\end{align}
Therefore, the verification operator for this strategy is given by
\begin{align}\label{eq:verification-strategy-GHZ-W}
    \Omega_\mu := \mu(Z) M_Z + \mu(X) \widehat{M}_X, 
\end{align}
where $\mu$ is a probability distribution satisfying $\sum_P\mu(P)=1$. 
The whole procedure is illustrated in Fig.~\ref{fig:GHZ-W verification}.

\begin{figure}[!htbp]
    \centering
    \includegraphics[width=0.6\linewidth]{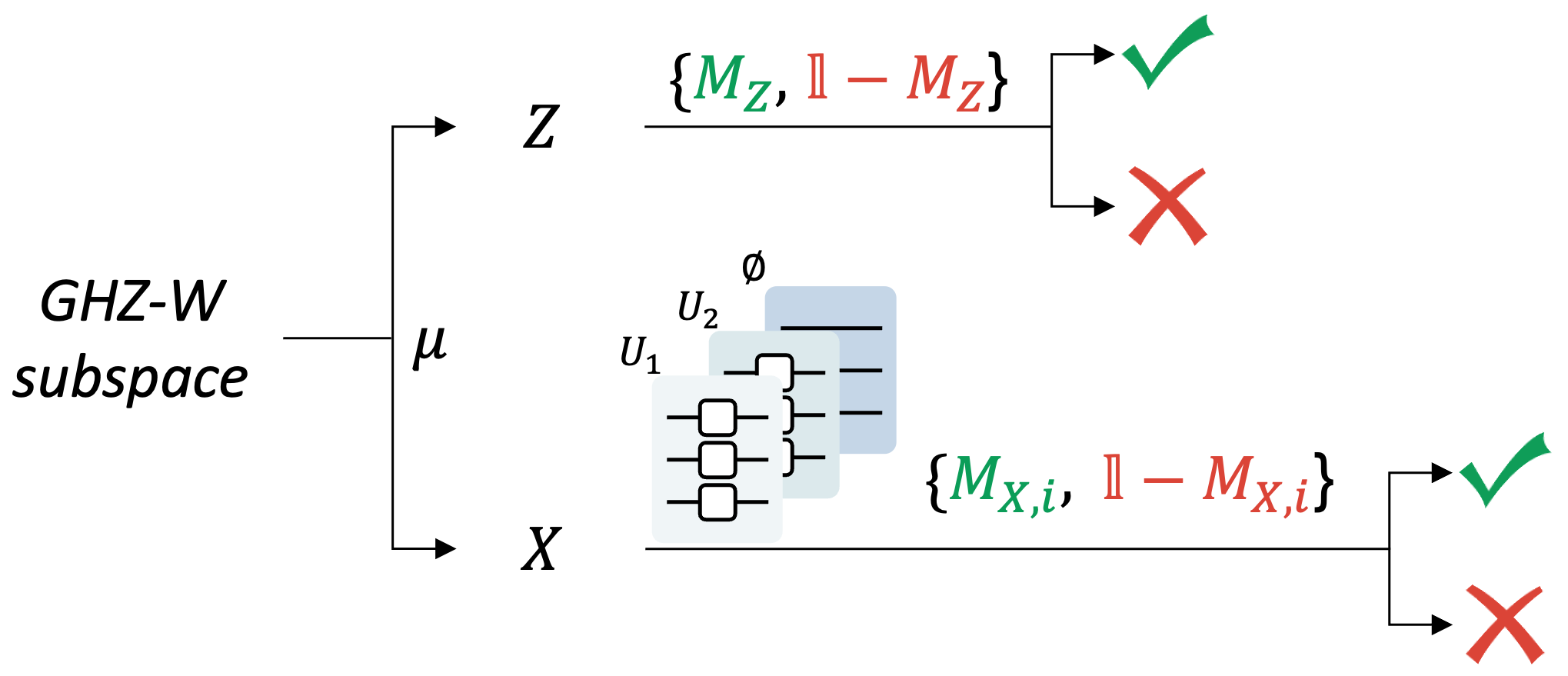}
    \caption{\raggedright 
    The one-way adaptive verification strategy for GHZ-W subspace. 
    We first randomly select $P\in\{X, Z\}$ according to a predefined probability distribution $\mu$. 
    (1) If $P=Z$, we perform the test operator $M_Z$. 
    (2) If $P=X$, we apply a unitary gate randomly chosen from the set $\{U_1, U_2, \emptyset\}$, where $\emptyset$ denotes applying no gate.
    Then, we perform the test operator $M_{X,i}$ starting with an $X$ measurement on the $i$-th qubit. 
    Based on its measurement result, we proceed with a two-qubit verification strategy on the remaining qubits.}
    \label{fig:GHZ-W verification}
\end{figure}

\paragraph*{Performance analysis.}
Obviously, the choice of $\mu(P)$ influences the performance of the strategy. 
Fortunately, the optimal probability distribution can be determined analytically, as shown in the following proposition. 

\begin{proposition}
\label{lem: spectral gap of rotation strategy}
The strategy operator defined in Eq.~\eqref{eq:verification-strategy-GHZ-W}, achieves the largest spectral gap of $141/317\approx 0.445$ when $\mu^\star(X) = 240/317 \approx 0.757$.
\end{proposition}

\begin{proof}
The analysis of the spectral gap relies on the invariant properties of the subspace $\cV_3$. 
Suppose $M$ is a test operator for the subspace $\cV_3$.
Define 
\begin{align}
    M' &:= \frac{1}{3} \left(M + U_1^\dagger M U_ 1 + U_2^\dagger M U_2\right).  
\end{align}
With the fact that $\cV_3$ is invariant under $U_i$, each term $U_i M U_ i^\dagger$ is also a valid test operator. 
Then, incorporating qubit permutations, the averaged operator of $M$ can be defined as
\begin{align}
    \overline{M} 
    := \frac{1}{6} \sum_{\pi\in \bS_3} V_{\pi} M'  V_{\pi}^\dagger
    = \left[\begin{matrix}
        a & 0 & 0 & 0 & 0 & 0 & 0 & b \\
        0 & d & e & 0 & e & 0 & 0 & 0 \\
        0 & e & d & 0 & e & 0 & 0 & 0 \\
        0 & 0 & 0 & f & 0 & g & g & 0 \\
        0 & e & e & 0 & d & 0 & 0 & 0 \\
        0 & 0 & 0 & g & 0 & f & g & 0 \\
        0 & 0 & 0 & g & 0 & g & f & 0 \\
        b & 0 & 0 & 0 & 0 & 0 & 0 & c 
    \end{matrix}\right], 
\end{align}
where $a,b,c,d,e,f,g$ are coefficients. 
As $M$ is a test operator, the states $\ket{{\rm GHZ}}$ and $\ket{{\rm W}}$ are eigenstates of $\overline{M}$, i.e., 
\begin{align}
    \begin{cases}
        \overline{M} \ket{\rm GHZ} = \ket{\rm GHZ} \\
        \overline{M} \ket{\rm W} = \ket{\rm W}
    \end{cases}
    \quad \Rightarrow \quad 
    \begin{cases}
        d = 1-2e \\
        b = 1 - a \\
        c = a
    \end{cases}. 
\end{align}
Therefore, $\overline{M}$ reduces to the form 
\begin{align}
    \overline{M} 
    = \left[\begin{matrix}
        a & 0 & 0 & 0 & 0 & 0 & 0 & 1-a \\
        0 & 1-2e & e & 0 & e & 0 & 0 & 0 \\
        0 & e & 1-2e & 0 & e & 0 & 0 & 0 \\
        0 & 0 & 0 & f & 0 & g & g & 0 \\
        0 & e & e & 0 & 1-2e & 0 & 0 & 0 \\
        0 & 0 & 0 & g & 0 & f & g & 0 \\
        0 & 0 & 0 & g & 0 & g & f & 0 \\
        1-a & 0 & 0 & 0 & 0 & 0 & 0 & a 
    \end{matrix}\right]. 
\end{align}
In addition to $\ket{{\rm GHZ}}$ and $\ket{{\rm W}}$, the other eigenstates and eigenvalues are: 
\begin{align}
    \ket{v_1} &= \frac{1}{\sqrt{2}} (\ket{000} - \ket{111}), \quad 
    \lambda_1 = 2a - 1, \\
    \ket{v_2} &= \frac{1}{\sqrt{2}} (\ket{001} - \ket{010}), \quad 
    \lambda_2 = 1 - 3e, \\
    \ket{v_3} &= \frac{1}{\sqrt{2}} (\ket{001} - \ket{100}), \quad 
    \lambda_3 = 1 - 3e, \\
    \ket{v_4} &= \frac{1}{\sqrt{2}} (\ket{011} - \ket{101}), \quad 
    \lambda_4 = f - g, \\
    \ket{v_5} &= \frac{1}{\sqrt{2}} (\ket{011} - \ket{110}), \quad 
    \lambda_5 = f - g, \\
    \ket{v_6} &= \frac{1}{\sqrt{3}} (\ket{011} + \ket{101} + \ket{110}) \quad  
    \lambda_6 = f + 2g. 
\end{align}
Obviously, the spectral gap of $\overline{M}$ is given by 
\begin{align}
    \nu(\overline{M}) 
    &= 1 - \max\{2a-1, 1-3e, f-g, f+2g\}. 
\end{align}
Now consider the verification operator $\Omega_{\mu}$, defined in Eq.~\eqref{eq:verification-strategy-GHZ-W}. 
With the definitions of $M_Z$ and $\widehat{M}_X$, we have 
\begin{align}
    M_Z &= \overline{M}_Z, \\
    \widehat{M}_X &= \overline{M}_X. 
\end{align}
Therefore, the spectral gap of $\Omega_{\mu}$ is given by 
\begin{align}
    \nu(\Omega_\mu) 
    &= 1 - \max\left\{
        1 - \frac{13}{20}\mu(X),
        1 - \frac{47}{80}\mu(X),
        \frac{131}{240} \mu(X),
        \frac{11}{15}\mu(X)
        \right\} \\
    &= \min\left\{
    \frac{47}{80}\mu(X), 1 - \frac{11}{15}\mu(X)\right\}. 
\end{align}
Thus, when $\mu^\star(X) = 240/317 \approx 0.757$, 
the spectral gap reaches its maximum value:
\begin{align}
    \nu(\Omega_{\mu^\star}) 
    = \frac{141}{317}
    \approx 0.445.
\end{align}
\end{proof}
We also compare our analytical results in Lemma~\ref{lem: spectral gap of rotation strategy} with the numerical results, 
as shown in Fig.~\ref{fig:numerical results}(b). 
We sample $\mu(X)$ from $0$ to $1$ with a step size of $0.05$ and find that the results are consistent.
Therefore, to achieve a confidence level of $1-\delta$, the required number of state copies is given by
\begin{align}\label{eq:rotation-strategy}
    N(\Omega_{\mu^\star}) 
    = \left\lceil\frac{317}{141}\times \frac{1}{\epsilon}\ln\frac{1}{\delta}\right\rceil 
    \approx \left\lceil2.248\times\frac{1}{\epsilon}\ln\frac{1}{\delta}\right\rceil.
\end{align}

\emph{Remark.} In each round, we can also select a measurement $P\in\{X, Y, Z\}$ according to some probability distribution. 
However, the numerical results shows that the optimal spectral gap is achieved when the probability of performing $Y$ measurement is $0$. Thus, it is sufficient to consider the rotation strategy.

\subsection{Comparisons}

As mentioned previously, to achieve a confidence level of $1 - \delta$, the globally optimal verification strategy requires only $\epsilon^{-1}\ln\delta^{-1}$ state copies, but it involves entangled measurements. 
In the previous subsections, we proposed two verification strategies based on one-way adaptive measurements: the XZ strategy and the rotation strategy.
The XZ strategy requires a total of four test operators, while the rotation strategy requires ten test operators. 
However, the rotation strategy requires approximately $2.248\epsilon^{-1} \ln\delta^{-1}$ state copies, which is fewer than that of the XZ strategy. 
In Fig.~\ref{fig:comparison}, we compare the efficiency of these strategies. 
We set $\delta=0.001$ and adjust $\epsilon$ from $0.001$ to $0.1$. 
Each line represents the minimum number of state copies required to achieve a confidence level of $1-\delta$.
The globally optimal verification strategy is the most efficient, with the rotation strategy surpassing the XZ strategy in efficiency. However, executing the globally optimal strategy is experimentally challenging. Practically, the same confidence level is attainable with a few local measurement settings, 
requiring about double the state copies compared to the global approach.

\begin{figure}[!htbp]
    \centering
    \includegraphics[width=0.6\linewidth]{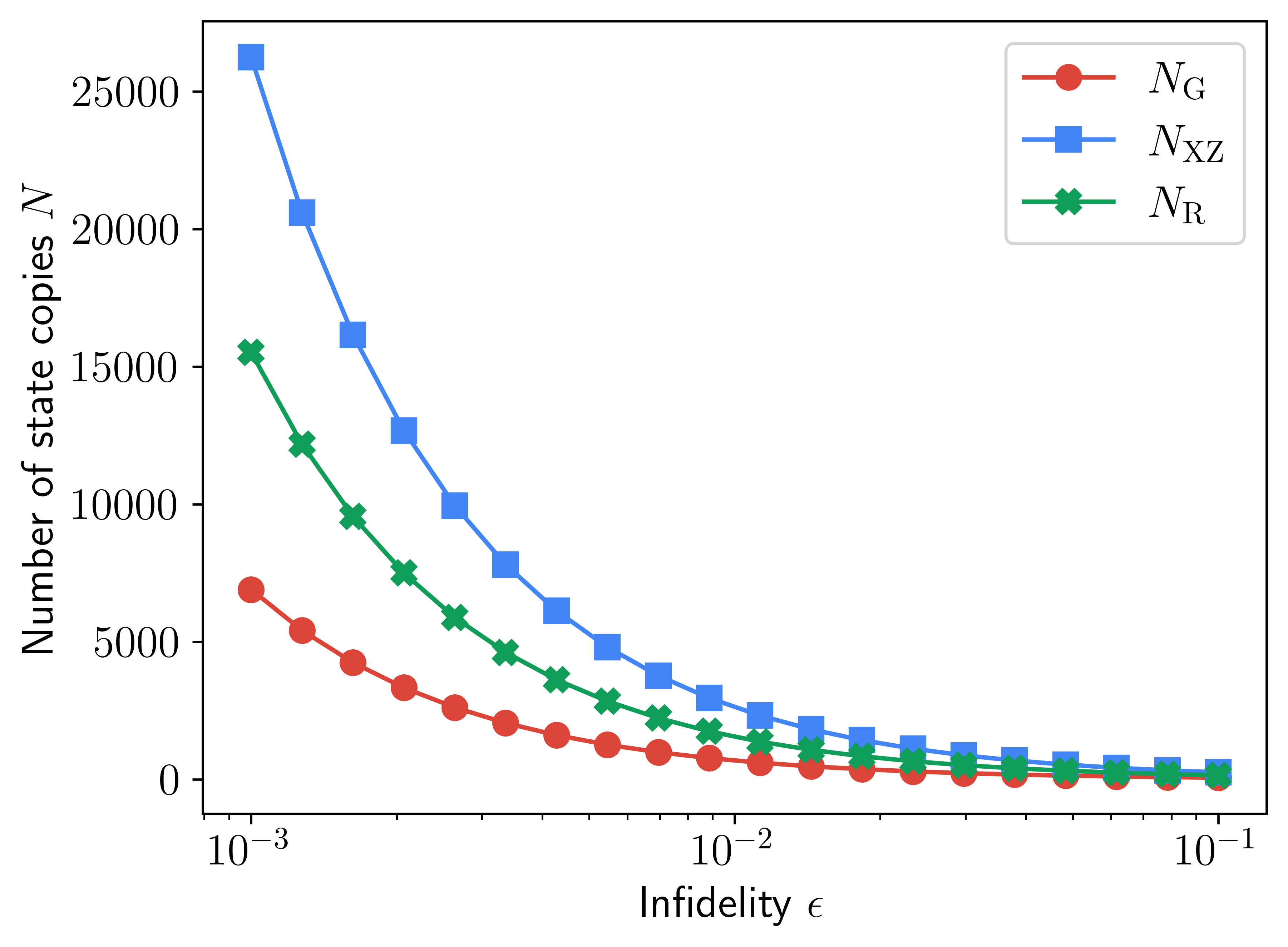}
    \caption{\raggedright 
    Comparison of the total number of state copies required to verify the three-qubit GHZ-W subspace for 
    different strategies as a function of the infidelity $\epsilon$, where $\delta = 0.001$.
    Here, $N_{\rm G}$ is the sample complexity of the globally optimal strategy 
    given in Eq.~\eqref{eq:global-strategy},
    $N_{\rm XZ}$ is the sample complexity of the XZ strategy
    given in Eq.~\eqref{eq:XZ-strategy}, and 
    $N_{\rm R}$ is the sample complexity of the rotation strategy
    given in Eq.~\eqref{eq:rotation-strategy}.}%
    \label{fig:comparison}
\end{figure}

\section{Conclusions}

In this work, we investigated the task of verifying the three-qubit GHZ-W genuinely entangled subspace 
using adaptive local measurements. 
By exploiting the local symmetry properties of the GHZ-W subspace, 
we first designed ten test operators and then constructed two efficient verification strategies: the XZ strategy and the rotation strategy.
The XZ strategy, employing four test operators, requires approximately $3.817\epsilon^{-1}\ln\delta^{-1}$ state copies to achieve a confidence level of $1-\delta$. 
In contrast, the rotation strategy, utilizing all ten test operators, achieves the same confidence level with a reduced sample complexity of $2.248\epsilon^{-1}\ln\delta^{-1}$.
Notably, the sample complexity of the rotation strategy is only approximately twice that of the globally optimal verification strategy, demonstrating its high efficiency. 
Along the way, we comprehensively analyzed the two-dimensional two-qubit subspaces, 
classifying them into three distinct types: unverifiable, verifiable, and perfectly verifiable subspaces. 
Interestingly, we demonstrated the existence of two-qubit entangled subspaces that are inherently unverifiable with local measurements, 
highlighting fundamental limitations in local entanglement verification.

Our findings raise several important open questions. 
A primary challenge lies in formulating and rigorously proving \emph{optimal} verification strategies with local measurements for arbitrary subspaces. 
Moreover, extending our approach to larger GHZ-W subspaces and other entangled subspaces remains an open area of research.

\section*{Acknowledgements}

This work was supported by
the National Key R\&D Program of China (No. 2022YFF0712800),
the National Natural Science Foundation of China (No. 62471126), 
the Jiangsu Key R\&D Program Project (No. BE2023011-2), 
the SEU Innovation Capability Enhancement Plan for Doctoral Students (No. CXJH\_SEU 24083), 
and the Fundamental Research Funds for the Central Universities (No. 2242022k60001).


\end{document}